\documentclass[10pt]{article}
\usepackage[margin=1in]{geometry}
\usepackage{authblk}
\usepackage{multicol}
\usepackage{hhline}

\usepackage{graphicx}
\usepackage{amsmath}
\usepackage{amssymb}
\usepackage{mathtools}
\usepackage{dsfont}
\usepackage{bm}
\usepackage{color}
\usepackage{appendix}
\usepackage{comment}

\usepackage{amsthm}

\usepackage{cite}

\usepackage[normalem]{ulem}

\input{Qcircuit}

\newcommand{\complex}{{\mathbb C}}
\newcommand{\reals}{{\mathbb R}}

\newtheorem{theorem}{Theorem}
\newtheorem{rem}{Remark}

\newtheorem{prop}{Proposition}
\newtheorem{cor}{Corollary}
\newtheorem{example}{Example}

\begin{document}

\title{On the Representation of Boolean and Real Functions as Hamiltonians for Quantum Computing}

\author{Stuart Hadfield\footnote{email: stuart.hadfield@nasa.gov / shadfield@usra.edu}}
\affil{Quantum Artificial Intelligence Lab, %(QuAIL), 
NASA Ames Research Center, Moffett Field, CA 94035}
\affil{USRA Research Institute for Advanced Computer Science, % (RIACS), 
Mountain View, CA 94043}
\affil{Department of Computer Science, Columbia University, New York, NY 10027}

\maketitle
 
\begin{abstract}
 Mapping functions on bits to Hamiltonians acting on qubits has many applications in quantum computing. In particular, Hamiltonians representing Boolean functions 
are required for applications of quantum annealing or the quantum approximate optimization algorithm to combinatorial optimization problems. We show how 
such functions are naturally represented by Hamiltonians given as sums of Pauli~$Z$ operators (Ising spin operators) with the terms of the sum corresponding to the function's Fourier expansion.
For many classes of Boolean functions which are given by a compact description, such as a Boolean formula in conjunctive normal form that gives an instance of the satisfiability problem, it 
is \#P-hard to compute its Hamiltonian representation,
i.e., as hard as computing its number of satisfying assignments. 
On the other hand, no such difficulty exists generally for constructing Hamiltonians representing a real function such as a sum of local Boolean clauses each acting on a fixed number of bits as is common in constraint satisfaction problems. 
We show composition rules for explicitly constructing Hamiltonians representing a wide variety of
Boolean and real functions by combining Hamiltonians representing simpler clauses as building blocks, which %
are particularly suitable for direct implementation %
as classical software. 
We further apply our results to the construction of controlled-unitary operators, and 
to the special case of operators that compute function values in an ancilla qubit register.   
Finally, we outline several additional applications and extensions of our results to
quantum algorithms for optimization.  
A %
goal of this %
work is to provide a \textit{design toolkit for quantum optimization} 
which may be utilized by experts and practitioners alike in the 
construction and analysis of new quantum algorithms, and at the same time to %
provide a unified framework for the various 
 constructions appearing in the literature. 
\end{abstract}

\maketitle

\section{Introduction}
A %
basic %
requirement of many %
quantum algorithms 
is the ability to translate between mathematical functions acting on a domain, typically strings of bits, and  quantum mechanical Hamiltonian operators acting on qubits. 
In particular, 
mapping Boolean or real functions %
to %
Hamiltonians has important 
applications %
in quantum algorithms 
and heuristics 
for solving decision %
or optimization problems 
such as %
quantum annealing and adiabatic quantum optimization%
~\cite{nishimoriQA,farhi2000quantum,mcgeoch2014adiabatic},  %
or the quantum approximate optimization algorithm 
and quantum alternating operator ansatz
(QAOA) \cite{hogg2000quantum,Farhi2014,hadfield2019quantum}, or the related variational quantum eigensolver%
~\cite{peruzzo2014variational}. 
Explicit %
Hamiltonian %
constructions for 
the application of these algorithms to 
a variety of 
prototypical problems can be found in  \cite{hadfield2019quantum,LucasIsingNP}, though prior work has mostly focused on reductions to quadratic Hamiltonians at the expense of additional qubits, for example in the penalty term approach of quantum annealing; we explain how direct (not necessarily quadratic) mappings are desirable for quantum gate model algorithms. 
Such quantum algorithms are promising, in particular,   
as possible paths towards performing useful computation on near-term quantum computing devices. 
Indeed, decision and optimization problems are ubiquitous across science and engineering, yet often appear to be computationally difficult. Despite years of investigation, efficient algorithms %
often remain elusive %
\cite{Ausiello2012complexity,AroraBarak}.  
 Hence, 
the potential for new approaches to tackling these problems on quantum computers 
is an exciting %
development. 

Nevertheless, the conceptual barrier to entry to studying these quantum 
algorithms and providing new insights remains high, especially for practitioners in the domain where a given problem arises, %
who may not be familiar with quantum computing beyond the basics. 
It is thus important to develop tools and methodologies which are accessible to scientists and researchers from different domains, %
and are as independent of knowing the low-level details of quantum computing as possible, towards  
enabling easier cross-fertilization of different ideas and techniques. %
 At the same time, it is useful to provide a %
rigorous general %
foundation for %
existing constructions and tools %
found 
in the literature, 
often in a specific context. %
Thus, a motivating goal of this work %
is to provide a %
 \textit{design toolkit} 
of basic results and methodologies which can be used by experts or laymen alike to design, implement, and analyze  
quantum algorithms. %
To this end, our results %
allow for straightforward implementation 
as computer programs such that Hamiltonian and quantum circuit 
mappings for many classes of problems may be 
automatically generated. 

In this paper we show %
a general theory of mappings of Boolean and %
real functions to diagonal Hamiltonians acting on qubits, 
and give simple logical rules for the explicit construction of these Hamiltonians %
which include 
many common %
classes of functions such as Boolean formulas and circuits. We also address the question of when such Hamiltonians may or may not be constructed efficiently.
We show how our 
results may be applied %
to the construction of unitary operators controlled by Boolean predicates, which are used, %
for example, in several of the QAOA mixing operator constructions of~\cite{hadfield2019quantum}. 
Our %
results are general and give a methodical approach to 
derive many of the mappings in the literature such as those of 
\cite{hadfield2019quantum,LucasIsingNP}. 
We emphasize that our results have applications to quantum algorithms beyond %
quantum annealing or QAOA, and we 
discuss a number of examples.

We elaborate on our results, which we summarize in the next section. 
Consider a function $f$ acting on $n$ bits. 
We say a Hamiltonian $H_f$ \textit{represents} $f$ %
if it satisfies 
\begin{equation}  \label{eq:eigenvalueHf}
H_f \ket{x} = f(x) \ket{x}
\end{equation}
for each input string $x \in \{0,1\}^n$ with corresponding (i.e., encoded as) computational basis state~$\ket{x}$. (We will typically assume~$f$ is decidable and moreover can be efficiently evaluated classically; see for example~\cite{nielsen1997computable} for a discussion of Hamiltonian families that encode the halting problem.)
We show how arbitrary $n$-bit Boolean or real functions are naturally represented as diagonal Hamiltonians 
given by weighted sums of Pauli $Z$ operators,  %
with terms corresponding to the function's 
Fourier expansion, as summarized in  Theorem~\ref{thm:1booleanFourierRep} below. 
Such Hamiltonians generalize the Ising %
model of interacting spin-$1/2$ particles well-known in physics. %
Our results rely on the %
Fourier analysis 
of Boolean functions, which has %
a long history of important applications in computer science 
 \cite{kahn1988influence,linial1993constant,beigel1993polynomial,nisan1994degree,hammer2012boolean,gu2019majorana}
and %
in particular quantum computing~\cite{beals2001quantum,ambainis2006polynomial, montanaro2008quantum,boixo2017fourier}. 
We use our results to derive explicit Hamiltonian representations of a number of basic Boolean predicates shown in Table~\ref{tab:basicHams} below. 
Combining the rules of classical propositional logic with the properties of the Fourier expansion %
leads to %
composition rules for constructing Hamiltonians representing %
 conjunctions, disjunctions, exclusive or, and other functions of simpler clauses by combining their Hamiltonian representations in particular ways; see Theorem~\ref{thm:compositionRules} below.   
 Furthermore, these mappings directly extend to constructing Hamiltonians representing weighted sums of clauses, which are a primary class of Hamiltonians considered in quantum annealing and QAOA, for example, for %
 constraint satisfaction problems.

We also consider the computational complexity of such Hamiltonian constructions,
which naturally depends on how the function $f$ is provided as input. 
Many %
combinatorial properties of a given function can be \lq\lq read off\rq\rq\ from its Fourier coefficients \cite{od2014analysis}. 
This presents an obstruction to computing the %
Hamiltonian representation for general Boolean functions; we show that computing the identity component of $H_f = \widehat{f}(\emptyset) I + \dots $, which is given 
by the first Fourier coefficient $\widehat{f}(\emptyset)$ of $f$,
is as hard as counting the the number of inputs such that $f=1$,
which %
in general is computationally intractable. 
For example, if $f$ is a
Boolean formula on $n$ variables given in conjunctive normal formula form with a ${\rm poly(n)}$ size description, i.e., an instance of the satisfiability problem (SAT), then this task is $\#$P-hard and hence  
it is not believed possible to efficiently compute $\widehat{f}(\emptyset)$ in general; 
if such a classical polynomial-time algorithm existed 
then we would have P=NP. 
Hence, we cannot efficiently construct explicit Hamiltonian representations of many $n$-bit Boolean functions, %
even when such a function %
may be  efficiently evaluated. We apply our results to show a similar result holds for computing the Pauli operator expansion of quantum circuits that compute~$f$ in a register. %

Nevertheless, there is no such difficulty for \textit{local} Boolean functions $f_j$ %
where each $f_j$ acts on a constant number of bits. 
This allows us to efficiently construct Hamiltonians representing \textit{pseudo-Boolean} functions of the form $f(x) = \sum_{j=1}^m w_j f_j (x)$, with $w_j\in\reals$ and $m={\rm poly}(n)$. %
Such real functions (and their corresponding Hamiltonians) may be constructed to directly encode optimization problems of interest, or 
such that their minimum value arguments (ground state eigenvectors) encodes the solution to a corresponding decision problem; similar to, for example, how solving the MAX-SAT problem (finding the maximum possible number of satisfied clauses) also solves SAT. %
For a pseudo-Boolean function, %
its Fourier coefficients do not %
allow its extremal values to be \lq\lq read off\rq\rq\ in the same way and so its Hamiltonian representation can often be computed efficiently. 
Indeed, this is a common approach 
to encoding decision problems such as SAT into the framework of 
quantum annealing~\cite{LucasIsingNP}. Our results likewise apply to designing penalty term approaches for problems or encodings with hard feasibility constraints which we discuss in Section~\ref{sec:penalties}.

\renewcommand{\arraystretch}{1.25}  %
\begin{table}[h]
\begin{center}
\caption{Hamiltonians representing basic Boolean clauses.} 
\label{tab:basicHams}
\begin{tabular}{| c | c || c | c |}
	\hline
	$f(x)$ &$H_f$ & $f(x)$ & $H_f$  \\ %
	\hhline{|=|=||=|=|}
	$x$ & $\tfrac12 I - \tfrac12 Z$ & $\overline{x}$ &$\tfrac12 I + \tfrac12 Z$   \\
	\hline
			$x_1 \oplus x_2 $ & $ \tfrac12 I-\tfrac12 Z_1Z_2$ & 
	 $\bigoplus_{j=1}^kx_j$  & $ \frac12 I - \frac12Z_1Z_2 \dots Z_k$  \\
	\hline
		$x_1 \wedge x_2$   &  $\tfrac14I-\tfrac14(Z_1+Z_2-Z_1Z_2)$  & $\bigwedge_{j=1}^kx_j$ & $\tfrac{1}{2^k} \prod_j (I-Z_j)$ \\
	\hline
	$x_1 \vee x_2 $  & $\tfrac34I -\tfrac14(Z_1+Z_2+Z_1Z_2)$  & $\bigvee_{j=1}^kx_j$ & $I-\frac{1}{2^k} \prod_j (I+Z_j)$ \\
	\hline
	$\overline{x_1x_2} $ & $ \tfrac34 I+  \tfrac14( Z_1 + Z_2 - Z_1 Z_2)$ & 
	$x_1 \Rightarrow x_2$  & $ \frac34 I + \frac14(Z_1 - Z_2 + Z_1Z_2)$  \\
	\hline
\end{tabular}
\end{center}
\end{table}
\subsection{Main Results: A %
Toolkit for Quantum Optimization Algorithms}
We summarize our main results. The details of representing Boolean and real (pseudo-Boolean) functions as Hamiltonians  are addressed in Section~\ref{sec:DiagHams} where we first briefly review and then apply Fourier analysis on the Boolean cube. In Section~\ref{sec:diagApplications} we outline the application of our results to constructing diagonal operators for  optimization applications, including ground state logic approaches and the implementation of diagonal phase unitaries for QAOA. 
We extend our results to (non-diagonal) controlled unitaries in Section \ref{sec:nonDiagHams}, including the setting of operators computing Boolean functions in ancilla registers, which is complementary to our main Hamiltonian approach.
Several illustrative examples and remarks are provided throughout. 
We conclude with a discussion of our results and future research questions in Section~\ref{sec:discussion}. %

\subsubsection{Boolean Functions}
Boolean-valued %
functions %
on the Boolean cube~$\{0,1\}^n$ 
are uniquely represented as diagonal Hamiltonians as in (\ref{eq:eigenvalueHf}) by a %
real multilinear polynomial of Pauli $Z$ operators, with terms %
corresponding to the function's Fourier expansion. %

\begin{theorem}  \label{thm:1booleanFourierRep}
For a Boolean function $f:\{0,1\}^n \rightarrow \{0,1\}$, 
the unique %
Hamiltonian on $n$ qubits satisfying $\;H_f \ket{x} = f(x)\ket{x}\;$ 
for each computational basis state $\ket{x}$ is %
\begin{eqnarray}  \label{eq:Hfexpan}
H_f &=& \sum_{S\subset [n]} \widehat{f} (S) \; \prod_{j\in S} Z_j          %
\,\,=\,  \widehat{f}(\emptyset)I + \sum_{j=1}^{n}\widehat{f}(\{j\})Z_{j} +  \sum_{j<k}\widehat{f}(\{j,k\})Z_jZ_k + \dots%
\end{eqnarray}
where the Fourier coefficients 
\begin{equation}
    \widehat{f} (S)  = \frac{1}{2^n} \sum_{x\in \{0,1\}^n} f(x) (-1)^{S\cdot x}=\frac{1}{2^n} {\rm tr}( H_f \prod_{j\in S} Z_j)
\end{equation}
satisfy 
$ \widehat{f}(\emptyset) \in [0,1] $, $\widehat{f} (S) \in [-\frac12,\frac12 ]$
for $S \neq \emptyset$, 
\begin{equation}  \label{eq:thm2eq2}
  \sum_{S\subset [n]} \widehat{f} (S)  = f(0^n)
\end{equation}
where $0^n$ denotes the input bit string of all $0$s, and
\begin{equation}  \label{eq:f0coefficient}
 \sum_{S\subset [n]}  \widehat{f} (S)^2 = \; \frac{1}{2^n} \sum_{x\in \{0,1\}^n} f(x) \;  = \;  \widehat{f}(\emptyset).
\end{equation}
\end{theorem}
\noindent The proof of the theorem is shown in Sec.~\ref{sec:boolean}. %
Here we have used the %
notation 
$S\cdot x := \sum_{j\in S}x_j$, 
 $[n]:=\{1,2,\dots,n\}$, 
and $Z_j:=I\otimes \dots \otimes I \otimes Z \otimes I \dots \otimes I$
to denote the Pauli $Z$ operator applied to the $j$th qubit. 
We emphasize that the theorem also applies to functions %
that depend only on a subset of $k<n$ bits, where $k$ may be independent of $n$. For example, the function $x_1\oplus x_2\oplus x_3$ is easily seen to map to the Hamiltonian $\tfrac12 I-\tfrac12Z_1Z_2Z_3$. 
We note that expansions similar to (\ref{eq:Hfexpan}) %
have been considered in the context of implementing diagonal unitaries~\cite{schuch2003programmable,welch2014efficient} and analyzing their quantum gate complexity~\cite{amy2018controlled}, whereas Theorem~\ref{thm:1booleanFourierRep} and our results to follow are more generally applicable. 
The generalization of Fourier analysis 
to functions over complex numbers (i.e., quantum circuit amplitudes) has been employed to analyze %
the complexity of %
sampling from quantum circuit models such as IQP circuits~\cite{bremner2017achieving,boixo2017fourier}. 
We emphasize that in gate model applications, %
higher order Pauli $Z$ interactions may be implemented efficiently; this is in contrast to quantum annealing applications where interactions must typically be reduced to low-order (quadratic) ones using ancilliary qubits to accommodate physical implementation~\cite{LucasIsingNP}. %
See Sec.~\ref{sec:HamSim} for additional discussion.

Thus, from (\ref{eq:f0coefficient}) in we see that computing the Hamiltonian representation (\ref{eq:Hfexpan}) of a Boolean function is equivalent to computing its Fourier expansion, 
and is at least as computationally difficult as computing %
$\widehat{f}(\emptyset)$. 
\begin{cor}  \label{cor:sharpPFourier}
Computing the identity coefficient $\widehat{f}(\emptyset)$ of the Hamiltonian $H_f$ representing %
a Boolean satisfiability (SAT) formula $f$ %
(given in conjunctive normal form) 
is $\#P$-hard. 
Deciding if  $\widehat{f}(\emptyset)=0$ is equivalent to deciding if $f$ is unsatisfiable, in which case $H_f$ %
reduces to the $0$ matrix. 
\end{cor}

Hence, given a Boolean function $f$ %
such that counting its number of satisfying inputs is $\#P$-hard, computing the identity coefficient $\widehat{f}(\emptyset)$ of its Hamiltonian representation will be $\#P$-hard also. 
Such hard counting problems include not only functions 
corresponding to NP-hard decision problems such as SAT, but also certain functions corresponding to %
decision problems decidable in polynomial time, such as counting the number of perfect matchings in a bipartite graph; see, e.g.,~\cite{AroraBarak}. %
We emphasize that even if we %
can compute the value of each Fourier coefficient, 
a Hamiltonian~$H_f$ representing a general Boolean %
or real function on $n$ bits may require 
a number of Pauli~$Z$ terms that is exponentially large with respect to $n$; such an example is the logical AND of $n$ variables (see Table~\ref{tab:basicHams}).%

On the other hand, when a Boolean clause $f$ acts only a number of bits~$k< n$ that is constant or logarithmically scaling we may always efficiently construct its Hamiltonian representation as the number of nonzero terms in the  sum~(\ref{eq:Hfexpan}) (%
the \textit{size} of $H_f$) %
is in this case at most $2^k$. 
The \textit{degree} of $H_f$, $\deg(H_f)=\deg(f)=d$, is the maximum locality (number of qubits acted on) of any such term. 
Bounded-degree Hamiltonians $H_f = \sum_{S\subset [k], |S| \leq d} \widehat{f} (S) \; \prod_{j\in S} Z_j$, $k\leq n$, are similarly always efficiently representable\footnote{We say a family of functions $\{f_{n,j}\}$ on $n$ bits  %
is \textit{efficiently representable} as %
the Hamiltonians~$\{H_{f_{n,j}}\}$ if ${\rm size}(H_{f_{n,j}})$ grows polynomially with~$n$ } %
as we show 
${\rm size}(H_f)\leq  (e/d)^{d-1}k^d +1$ 
which is always 
polynomial in $n$ if $d=O(1)$.
For example, bounded-locality constraint satisfaction problems such as MAX-CUT or MAX-$\ell$-SAT (up to $\ell=O(\log n)$)  are described by sums of local clauses %
and hence are always
efficiently representable as Hamiltonians. %
We summarize mappings of some important basic clauses in Table~\ref{tab:basicHams} above.

\subsubsection{Constructing Hamiltonians using Propositional Logic}

We %
show 
formal rules for combining Hamiltonians representing different Boolean functions to obtain Hamiltonians representing more complicated logical expressions. %
In particular, we consider  %
the logical negation$(\neg)$, conjunction $(\wedge)$, disjunction $(\vee)$, exclusive or $(\oplus)$, and implication $(\Rightarrow)$ operations, and addition or multiplication by a real number when Boolean functions are embedded as a subset of real-valued functions.

\begin{theorem}[Composition rules]  \label{thm:compositionRules}
Let $f,g$ be Boolean functions represented by Hamiltonians  $H_f,H_g$. 
Then the Hamiltonians representing basic operations on $f$ and $g$ are given by %
  \begin{multicols}{2}
    \begin{itemize}
        \item $\,H_{\neg f} = H_{\overline{f}} = I - H_f$
        \item $\,H_{f\wedge g} = H_{fg} = H_f H_g$
        \item $\,H_{f\oplus g} =  H_f + H_g - 2H_fH_g$
            \item $\,H_{f\Rightarrow g} = I - H_f + H_fH_g$
        \item $\,H_{f\vee g} = H_f + H_g - H_fH_g$
           \item $\,H_{af+bg} = a H_f+b H_g\;\;\;\;  a,b \in \reals.$
    \end{itemize}
  \end{multicols}
\end{theorem}
The proof of theorem is given in Section~\ref{sec:compRules}. 
Hamiltonians for a %
wide variety of functions can be easily constructed using the  composition rules and results for basic clauses as in Table~\ref{tab:basicHams}, 
in particular %
Boolean formulas and circuits. 

\begin{rem}
Theorems~\ref{thm:1booleanFourierRep}, \ref{thm:compositionRules}, and \ref{thm:pseudoBool} below facilitate  straightforward software implementation for generating diagonal Hamiltonians, for example, automating %
QAOA mappings for constraint satisfaction problems with increasingly general types of constraints. %
Examples of three variable clauses are generated in Table~\ref{tab:3qubitHams} below. Such Hamiltonians may be useful \emph{intermediate representations} in applications.  
\end{rem}

\begin{table}[h]
\begin{center}
\caption{Example: Hamiltonians representing Boolean clauses on three variables. }
\label{tab:3qubitHams}
\begin{tabular}{| c | c |}
	\hline
	$f(x)$ &$H_f$   \\ %
\hhline{|=|=|}
    $MAJ(x_1,x_2,x_3) $ & $ \frac12 I - \frac14 \left(Z_1 + Z_2 + Z_3 - Z_1 Z_2 Z_3\right) $ \\
    \hline  	
		$NAE(x_1,x_2,x_3) $ & $ \frac34 I - \frac14(Z_1Z_2 + Z_1Z_3 + Z_2Z_3)$ \\
	\hline
	$MOD_3(x_1,x_2,x_3)$ & $\frac14 I+ \frac14(Z_1Z_2+Z_2Z_3+Z_1Z_3)$\\
			\hline
	 $1in3(x_1,x_2,x_3)$ &
       $\frac{1}{8}(3I +Z_1 +Z_2 +Z_3 -Z_1Z_2 -Z_2Z_3 -Z_1Z_3 -3Z_1Z_2Z_3)$ \\  
	\hline
\end{tabular}
\end{center}
\end{table}

\subsubsection{Pseudo-Boolean functions}
Consider a 
real function $f$ on $n$ bits %
given as a weighted sum of %
Boolean functions~$f_j$, 
$$ f(x) = \sum_{j=1}^m w_j f_j (x) \;\;\;\;\;  w_j \in \reals,$$  %
where each $f_j$ acts on a subset of the $n$ bits, and in the applications we consider often  %
$m=\textrm{poly}(n)$. 
Objective functions for constraint satisfaction problems, considered for example in QAOA, are %
often expressed in this form, with each $f_j$ given by a Boolean clause. 
A different example is the penalty term approach of quantum annealing, 
where the objective function is augmented with a number of high-weight 
penalty terms which perform (typically, local) checks that a state is valid; see Sec.~\ref{sec:penalties} for a discussion. 
For such \textit{pseudo-Boolean functions} we have the following result 
generalizing Thm.~\ref{thm:1booleanFourierRep}, which is  
shown in Section~\ref{sec:pseudoBool}. %
Pseudo-Boolean optimization is a rich topic and we refer the reader to \cite{boros2002pseudo} for an overview.

\begin{theorem} \label{thm:pseudoBool}
For an $n$-bit %
real function $f:\{0,1\}^n \rightarrow \reals$ %
the unique Hamiltonian on $n$ qubits satisfying $H_f \ket{x} = f(x)\ket{x}$ %
is 
\begin{eqnarray}  \label{eq:Hfexpan2}
 H_f = \sum_{S\subset [n]} \widehat{f} (S) \; \prod_{j\in S} Z_j,
\end{eqnarray} 
with coefficients $\widehat{f} (S) =
  \frac{1}{2^n} \sum_{x\in \{0,1\}^n} f(x) (-1)^{S\cdot x} =\frac1{2^n}{\rm tr}( H_f \prod_{j\in S} Z_j)\, \in \reals. $

In particular, for a pseudo-Boolean %
function %
$ f(x) = \sum_{j=1}^m w_j f_j (x) $, $w_j \in \reals$, 
where the $f_j$ are Boolean functions
corresponding to Hamiltonians $H_{f_j}$ as in~(\ref{eq:Hfexpan}), we have 
\begin{eqnarray}
H_f = \sum_{j=1}^m w_j H_{f_j} \;\;\;\;\; \text{ and } \;\;\;\;\; \widehat{f} (S)  = \sum_j \widehat{f_j} (S) \in 
 \reals
\end{eqnarray}
with $d:= \deg(H_f) \leq  \max_j \deg(f_j)$ and 
$\;{\rm size}(H_f)\leq \min\{ \sum_j {\rm size}(H_{f_j})
, (e/d)^{d-1}n^d +1\}$.
\end{theorem}

Our results also yield direct cost estimates for quantum simulation of diagonal Hamiltonians, which we elaborate on in Section~\ref{sec:HamSim}, though we emphasize that simulation %
is 
not the %
motivating application %
our results; %
 nevertheless, it is an important piece of our design toolkit. Indeed, such operators occur for example in QAOA or Grover's algorithm.  %
As the terms in~(\ref{eq:Hfexpan2}) mutually commute, we can \textit{simulate} such an a Hamiltonian, i.e., implement the operator $U=exp(-iH_ft)$ for some fixed  $t \in \reals$, 
using $O({\rm deg}(H_{f}) \cdot {\rm size}(H_{f}))$ many basic quantum gates.  %
Indeed, similar constructions for implementing diagonal unitaries have been previously proposed~\cite{schuch2003programmable,welch2014efficient}, though more efficient circuits often result by utilizing ancilla qubits such as the operators we consider in Proposition~\ref{prop:computeReg} below~\cite{barenco1995elementary,childs2004quantum}. %

The remaining two items of the section demonstrate how our results may be applied to %
the construction of more general (non-diagonal) %
Hamiltonians and unitary operators.

\subsubsection{Controlled Hamiltonians and Unitaries}
In many applications we require controlled Hamiltonian simulations. %
Such operators are closely related to \textit{block encodings} common in quantum computing~\cite{chakraborty2018power}. 
Consider two quantum registers of $k+n$ qubits. 
Given a Boolean function~$f(y)$ acting on $k$ bits and a unitary operator $U$ acting on $n$ qubits, 
we define the $(k+n)$-qubit $f$-controlled unitary operator $\Lambda_f (U)$ by its action on computational basis states 
$$ 
\Lambda_f (U) \ket{y}\ket{x} =  \left\{
                \begin{array}{ll}
                  \ket{y}\ket{x}\;\;\;\;\;\;\; \;  f(y)= 0  \\
                   \ket{y}U\ket{x} \;\;\;\;\;  f(y)= 1.
                \end{array}
              \right.
$$
Equivalently, 
as $H_f + H_{\overline{f}}=I$ we have the useful decomposition 
  $\Lambda_f (U)= H_f \otimes U + H_{\overline{f}} \otimes I$.

If $U$ is self-adjoint, then $\Lambda_f (U)$  %
is also a Hamiltonian. 
When $U$ is given as a time evolution under a Hamiltonian~$H$ for a time $t$,  
we have the following. 

\begin{prop}  \label{prop:controlledHams}
Let $f$ be a %
Boolean function represented by a $k$-qubit Hamiltonian $H_f$, 
and let $H$ be an arbitrary Hamiltonian acting on $n$ disjoint qubits. 
Then the $(k+n)$-qubit Hamiltonian 
\begin{equation}  \label{eq:ctrlHam}
\widetilde{H}_f = H_f \otimes H
\end{equation}
corresponds to %
$f$-controlled evolution under $H$,
i.e., for $t\in\reals$ satisfies %
\begin{equation}
    e^{-i \widetilde{H}_ft} = \Lambda_f ( e^{-iH t}).
\end{equation}
\end{prop}

The proof follows from
from exponentiating (\ref{eq:ctrlHam}) directly, see Section~\ref{sec:ControlledHams}. We emphasize the Proposition may be applied to generic control functions beyond the AND functions commonly considered in the literature (e.g., multi-controlled Toffoli gates).

\begin{rem}
[Advanced mixing operators for QAOA]
Proposition~\ref{prop:controlledHams} %
may be applied together with the Theorems above  %
to design controlled mixing operators %
for %
QAOA 
mappings of %
various optimization problems with hard %
constraints, with the important property of restricting the quantum evolution %
to the subspace of feasible states~\cite{hadfield2017qaoaPMES,hadfield2019quantum}. In particular, several mixing operators $\Lambda_f (e^{-iH\alpha})$ proposed in  \cite{hadfield2019quantum} 
implement %
 evolution under a local mixing Hamiltonian $B$ %
 controlled by a Boolean function $f$, %
 where the control function checks %
that the mixing action %
on a given basis state   
 will maintain feasibility %
and %
 acts nontrivially %
 only %
 when this is the case. For example, for MaxIndependentSet on a given graph, the transverse-field (bit-flip) mixer is applied for each vertex controlled by the variables corresponding to the its neighbors in the graph; see \cite{hadfield2019quantum} for details. %
 \end{rem}
 
We next consider 
the special case where the target Hamiltonian $H$ corresponds to a bit flip such that %
each computational basis state 
$\ket{y}\ket{0}$ is mapped to $\ket{y}\ket{f(y)}$.
 
\subsubsection{Computing Boolean Functions in a Register}
We show how our results may be applied to %
construct explicit %
unitary operators which reversibly compute function values in an ancilla qubit register  (i.e., implement \textit{oracle queries}), as well as their corresponding Hamiltonians.  
We remark that a related approach for implementing such operators in a specific gate set has recently appeared~\cite{soeken2020quantum}.  
 Recall that in the computational basis, the Pauli $X$ operator acts as $X\ket{0}=\ket{1}$ and $X\ket{1}=\ket{0}$, i.e., as the \textit{bit-flip}, or NOT, operation. %
 
\begin{prop}     \label{prop:computeReg}
For an $n$-bit Boolean function $f$ represented by a Hamiltonian $H_f$, let $G_f$ be %
 the unitary self-adjoint operator on $n+1$ qubits which acts on computational basis states $\ket{x}\ket{a}$ as
\begin{equation}
G_f \ket{x}\ket{a}=\ket{x}\ket{a \oplus f(x) }
\end{equation} 
where $x\in\{0,1\}^n$ and $a\in\{0,1\}$. 
Then
\begin{equation} \label{eq:regfunc}
    G_f =  \Lambda_f(X)  = e^{-i \frac{\pi}{2} \,  H_f \otimes (X-I)},
\end{equation}
and $G_f=I$ if and only if $f$ is unsatisfiable. 
\end{prop}
The proof is given in Section~\ref{sec:computeReg}. 
As before it remains computationally hard to compute the Fourier expansion of such operators in general. 

\begin{cor} \label{cor2}
If $f$ is given as a SAT formula in conjunctive normal form, then  
it is \#P-hard to compute the identity coefficient $\widehat{g}(\emptyset)={\rm tr}(G_f)/2^{n+1}$
of $G_f$, %
and NP-hard to decide if $\;\widehat{g}(\emptyset)\neq 1$.  
\end{cor}

Equation (\ref{eq:regfunc}) shows how Hamiltonian simulation may be used to compute a function~$f$ in a register. In particular, as the Pauli terms in $H_f\otimes X$ and $H_f\otimes I$ mutually commute, $G_f$ can %
be implemented with $2 \cdot{\rm size}(H_f)$ many multiqubit Pauli rotations with locality up to $\deg(H_f)+1$.

\section{Representing $n$-bit Functions as Diagonal Qubit Hamiltonians}  \label{sec:DiagHams}
Many important problems, algorithms, and operators 
in physics and computer science 
naturally involve Boolean predicates. Indeed, the relationship between logical propositions and physical observables is foundational in quantum mechanics~\cite{birkhoff1936logic}. 
For the case of qubits, we show how the natural unique representation of a Boolean function as a Hamiltonian composed of spin operators (Pauli $Z$ matrices)  
follows %
from %
classical Fourier analysis; see \cite{de2008brief,od2014analysis} for overviews of the subject.
We use these tools 
to extend our results to Hamiltonians representing more general functions built from sums, conjunctions, disjunctions, and other basic combinations of simpler Boolean clauses. 

\subsection{Boolean Functions}  \label{sec:boolean}
The class of Boolean functions on $n$ bits is defined as  $\mathcal{B}_n:=\{ f:\{0,1\}^n\rightarrow \{0,1\} \}$. %
Taken as the elements of a real vector space, %
for each $n$ they give a basis for the real functions $\mathcal{R}_n=\{f:\{0,1\}^n\rightarrow \reals \}$ on $n$ bits. 
Moreover, $\mathcal{R}_n$ is isomorphic to the vector space of diagonal Hamiltonians %
acting on $n$ qubits, or, equivalently, the space of $2^n \times 2^n$ diagonal real matrices. 
Thus, diagonal %
Hamiltonians naturally encode large classes of %
functions.

We say a Hamiltonian \textit{represents} a function $f$ if in the computational basis 
it acts as the corresponding multiplication operator, 
i.e., it satisfies the $2^n$ eigenvalue equations 
\begin{equation}   \label{eq:eigenvalueHf2}
    \forall x \in \{0,1\}^n\;\;\;\; H_f\ket{x}=f(x)\ket{x}.
\end{equation}
On $n$ qubits, this condition specifies $H_f$ uniquely (up to the choice of the computational basis). 
Equivalently, we may write $H_f = \sum_x f(x)\ket{x}\bra{x}$, which %
for a Boolean function  
$f \in \mathcal{B}_n$ 
becomes
 \begin{equation}   \label{eq:projector}
H_f = \sum_{x: f(x) =1} \ket{x}\bra{x}.
 \end{equation}
As Boolean functions satisfy $f^2=f$ we have $H_f^2 = H_f$, so $H_f$ is a projector\footnote{Projectors give quantum observables. In particular, for an arbitrary normalized $n$-qubit state $\ket{\psi}$, the probability $p_1$ of a computational basis measurement %
returning a satisfying bit string (i.e., an $x$ such that $f(x)=1$) is given by
$p_1 = \bra{\psi} H_f \ket{\psi}, $
i.e., is equal to the expected value of repeated measurements of %
$H_f$ on the state $\ket{\psi}$.}
 of rank $r=\#f := |\{x:f(x)=1\}| = \sum_x f(x)$.
Hence, 
the Hamiltonian $H_f$ for a Boolean function $f$ is equivalent to %
the projector onto the subspace spanned by basis vectors $\ket{x}$ such that $f(x)=1$, and 
given an~$f$ %
such a projector may be constructed using our results below. Hence,  
determining if $f$ is satisfiable is equivalent to determining if $H_f$ is not identically $0$, and determining $H_f$ explicitly in the form of (\ref{eq:projector}) is 
at least 
as hard as counting the number of satisfying assignments of $f$, or equivalently, computing $r={\rm rank}(H_f)$. 

We consider the standard computational basis of eigenstates of Pauli $Z$ operators  (often written as $\sigma_z$), %
defined by the relations $Z\ket{0} = \ket{0}$ and $Z\ket{1} = -\ket{1}$. 
We use $Z_j= I \otimes \dots I \otimes Z \otimes I \dots \otimes I$ to denote 
$Z$ acting on the $j$th qubit. 
Products of $Z_j$ over a set of qubits act as 
\begin{equation}
\prod_{j\in S} Z_j\ket{x} = \chi_S(x)\ket{x},
\end{equation}
where 
each \textit{parity function} $\chi_S (x): \{0,1\}^n \rightarrow \{-1,+1\}$ gives the parity of the bits of $x$ in the subset $S\subset[n]$, i.e., 
is $+1$ if and only if the number of bits of $x$ set to $1$ is even.  
Identifying each $S$ with its characteristic vector $S\in \{0,1\}^n$
 such that %
 $S\cdot x = \sum_{j \in S} \, x_j$,
we have 
\begin{equation}  \label{eq:parityfcn}
\chi_S (x) = (-1)^{S\cdot x}  = (-1)^{ \oplus_{j \in S} x_j}.
\end{equation}
Thus, each Hamiltonian $Z_S:=\prod_{j\in S} Z_j$ represents the function $\chi_S(x)$ in the sense of (\ref{eq:eigenvalueHf2}).

The set of parity functions on $n$ bits $\{ \chi_S (x): S \subset [n]\}$ also gives a %
basis for the real functions~$\mathcal{R}_n$. 
This basis is orthonormal 
with respect to the inner product defined by %
\begin{equation}     \label{eq:innerProduct}
 \langle f,g \rangle := \frac{1}{2^n} \sum_{x \in \{0,1\}^n} f(x) g(x) .
\end{equation}
Hence, every Boolean function $f\in \mathcal{B}_n$ may be written uniquely as  
\begin{equation}   \label{eq:FourierExpan}
f(x) = \sum_{S\subset [n]} \widehat{f} (S) \, \chi_S (x),
\end{equation}
called the \textit{Fourier expansion} (or, sometimes, \textit{Walsh} or \textit{Hadamard} expansion~\cite{macwilliams1977theory}, or \textit{phase polynomial}~\cite{amy2018controlled})
with %
 \textit{Fourier coefficients} given by the inner products of $f$ with the parity functions 
\begin{equation}   \label{eq:BooleanFourierCoefficients}
\widehat{f} (S) = \frac{1}{2^n} \sum_{x\in \{0,1\}^n} f(x) \chi_{S}(x) = \langle f, \chi_S \rangle,
\end{equation}
and satisfying \textit{Parseval's identity} 
\begin{equation}   \label{eq:parseval}
\sum_{S\subset [n]}\widehat{f}(S)^2 = \frac1{2^n}\sum_x f(x)^2.
\end{equation}
The quantity $\sum_{S\neq \emptyset}  \widehat{f} (S)^2 =:{\rm var}(f)$ is often referred to as the \emph{variance} of $f$. 
For $\{0,1\}$-valued functions, we have $f^2=f$ which implies $\sum_{S\subset [n]} \widehat{f}(S)^2 = \frac1{2^n}\sum_x f(x)=\widehat{f}(\emptyset)=\widehat{f}(\emptyset)^2+{\rm var}(f)$.  

We refer to the mapping from $f(x)$ to $\widehat{f}(S)$ as the \textit{Fourier transform} of~$f$. %
The sparsity of $f$, denoted ${\rm spar}(f)$, is the number of non-zero coefficients $\widehat{f} (S)$. 
When ${\rm spar}(f)={\rm poly}(n)$ we refer to $f$ as  \textit{polynomially sparse}; in particular, polynomially sparse functions %
yield 
Hamiltonians of size ${\rm poly}(n)$.  
The \textit{degree} of $f$, denoted $\deg(f)$, is defined to be the largest $|S|$ such that $\widehat{f}(S)$ is nonzero. Note that if $f$ depends on only $k\leq n$ variables, then $\deg(f) \leq k$.

We summarize our results %
on the representation of Boolean functions in Theorem~\ref{thm:1booleanFourierRep} above.

\begin{proof}[Proof of Theorem~\ref{thm:1booleanFourierRep}]
 From the %
identification of $\chi_S$ with $Z_S= \bigotimes_{j\in S} Z_j=\prod_{j\in S} Z_j$,  %
from (\ref{eq:FourierExpan}) 
we see that each Boolean function $f$ 
is represented as a diagonal Hamiltonian 
by a linear combination of tensor products of~$Z_j$ operators. As all Pauli operators and their tensor products %
are traceless except for the identity operator, 
we have $\widehat{f} (S)  = {\rm tr}( H_f \prod_{j\in S} Z_j) /2^n$%
where %
\footnote{${\rm tr} (H)$ denotes the \textit{trace} of the matrix $H$, i.e., the (basis-independent) sum of its diagonal elements, %
}
The results %
(\ref{eq:thm2eq2}) and (\ref{eq:f0coefficient}) %
follow from Parseval's identity 
(\ref{eq:parseval}) using $f^2=f$. 
Recall we define the degree (sometimes called the \textit{Pauli weight}) of such a Hamiltonian $H_f$, $\deg(H_f)$, to be the largest number of qubits acted on by any term in this sum, %
and the size, ${\rm size}(H_f)$, to be the number of (nonzero) terms.
Clearly, we have $\deg(H_f)=\deg(f)=:d$, and simple counting 
and the bound $\binom{n}{d} \leq n^d (e/d)^d /e$ 
gives ${\rm size}(H_f)={\rm spar}(f)%
\leq d\binom{n}{d}\leq (e/d)^{d-1}n^d +1$.
\end{proof}

We emphasize that 
the Hamiltonian coefficients $\widehat{f}(S)$ depend only on the function values %
$f(x)$, and are independent of \textit{how} such a function may be represented as input (e.g.,  formula, circuit, truth table, etc.). 
Many typical compact representations of Boolean functions as computational input such as Boolean formulas %
or Boolean circuits can be directly transformed to Hamiltonians using the composition rules of Theorem \ref{thm:compositionRules} which we derive below. 

\begin{rem}
The identification of Boolean functions as real polynomials allows application of Theorem~\ref{thm:1booleanFourierRep} to different domains or targets. Changing the target from $\{0,1\}$ to $\{1,-1\}$ for a given $f(x)$ corresponds to the function $g(x) = 2f(x)-1$ with Hamiltonian coefficients $\widehat{g}(\emptyset)=1-2\widehat{f}(\emptyset)$ and 
$\widehat{g}(S)=-2\widehat{f}(S)$ for $S\neq \emptyset$, so this %
can change ${\rm size}(H_f)$ by at most $1$. 
(In this case there are several %
differences %
from the $\{0,1\}$ approach; %
for example, %
Parseval's identity trivially becomes $\sum_{S\subset [n]}\widehat{f}(S)^2 = 1$.) %

Similarly, indeed, the Fourier expansion (\ref{eq:FourierExpan}) itself corresponds to changing the domain from $\{0,1\}^n$ to $\{1,-1\}^n$ (i.e., the polynomial obtained %
replacing each $Z_j$ in (\ref{eq:Hfexpan2}) by the variable $z_j=(-1)^{x_i}\in\{-1,1\}$).
\end{rem}

On the other hand, Theorem \ref{thm:1booleanFourierRep} shows that computing the Hamiltonian representation of a Boolean function is as hard as %
computing its number of satisfying assignments, which is believed to be a computationally intractable problem in general~\cite{AroraBarak}.  %
We illustrate this explicitly in Corollary ~\ref{cor:sharpPFourier} above where we show computing $\widehat{f}(\emptyset)$ %
can be $\#$P-hard. 
Moreover, arbitrary Boolean functions may have size (sparsity) exponential in $n$, %
in which case, even if we know somehow its Hamiltonian representation, 
we cannot implement or simulate this Hamiltonian efficiently (with respect to $n$) with the usual direct approaches.

As %
explained, Hamiltonians representing pseudo-Boolean functions %
often avoid these difficulties; for example, constraint satisfaction problems with objective
function given as the sum of a number %
of \textit{local} clauses %
each clause acting on at
most $k=O(\log n)$ bits (e.g., Max-$k$-Sat), and so the Hamiltonians %
representing each clause have degree $O(\log n)$ and size $O({\rm poly}(n))$,
and hence can be efficiently constructed. 
Applying the well-known results of \cite[Thm. 1 \& 2]{nisan1994degree} 
for such functions  
to Theorem \ref{thm:1booleanFourierRep} 
immediately %
gives the following useful Hamiltonian degree bounds. 
\begin{cor}
For a function $f\in\mathcal{B}_n$ %
that depends only on $k\leq n$ variables,  
represented as a Hamiltonian $H_f$ acting on $n$ qubits,  
the degree of  $H_f$ satisfies 
\begin{equation}
k \geq D(f) \geq \deg(H_f) \geq \log_2 k -O(\log\log k),
\end{equation}
where $D(f)$ is the decision tree complexity of $f$ and $D(f) = O({\rm poly}(\deg(H_f)))$.
\end{cor}
We emphasize that further %
results 
from the literature concerning the Fourier analysis of Boolean functions %
may be similarly %
adapted to obtain properties of corresponding Hamiltonians. %
In particular additional useful %
details of the Fourier coefficients may be found in \cite{od2014analysis,de2008brief,amy2018controlled}.

\subsection{Basic Clauses and Logical Composition Rules} \label{sec:compRules}
Boolean functions arise in many different applications, but  
are often given %
in or easily reduced to a %
\textit{normal form}. 
For example, SAT formulas are given in conjunctive normal form. %
Many other %
normal forms exist such as disjunctive, algebraic ($\oplus$), minterm or maxterm, etc.~\cite{halmos2009boolean}. %
Note that while logically equivalent,
the different forms may be quite %
different for computational purposes. 
For each form there corresponds a notion of size (which directly relates to %
the number of bits needed to describe a function in such a form). 

The laws of classical propositional logic allow for convenient manipulation of a given Boolean expression between logically equivalent forms. 
Hence, given Hamiltonians representing basic Boolean functions, it is useful to have a methodical way to combine them to in order to construct new Hamiltonians representing their logical conjunctions (AND), disjunctions (OR), etc. %
This approach often allows for easier construction of such a Hamiltonian than by working with the Fourier expansion directly. 
First, consider Hamiltonians representing basic functions and variables.   
The trivial functions $f=1$ and $f=0$ (resp. \textit{true} and \textit{false}) are represented by 
the $2^n$-dimensional Hamiltonians $H_1 = I$ and $H_0 = 0$,
respectively. 
Using the identity $(-1)^{x} = 1 - 2x$, we %
represent the $j$th variable $x_j$, considered as a multiplication operator, as the Hamiltonian 
\begin{equation}
     H_{x_j} \,=\, I^{\otimes j-1} \otimes \ket{1_j}\bra{1_j} \otimes I^{\otimes n-j}\,=\, \frac12I - \frac12Z_j, 
\end{equation}
which acts according to the value of the $j$th bit as %
$H_{x_j}\ket{x}=x_j\ket{x}$. 
Similarly, the logical negation of the $j$th variable is represented as $H_{\overline{x}_j}=I/2+Z_j/2$. 
For clarity we will avoid writing tensor factors of identity operators explicitly when there is no ambiguity, and %
for convenience %
we sometimes write $x_j$ to mean the Hamiltonian~$H_{x_j}$, and likewise for other basic functions such as~$\overline{x}_j$. 

Applying the laws of propositional logical leads to the composition rules of Theorem~\ref{thm:compositionRules}. 
 
\begin{proof}[Proof of Thm. \ref{thm:compositionRules}]
The logical values $1$ and $0$ %
(i.e., \textit{true} and \textit{false}) are represented as the identity matrix $I$ and the zero matrix, respectively. 
Each result follows from the natural embedding 
of $f,g \in \mathcal{B}_n$ into $\mathcal{R}_n$, %
the real vector space of real functions on $n$ bits.  
From linearity of the Fourier transform, we immediately have 
$H_{af+bg} = aH_f + bH_g \text{ for } a,b \in \reals.$
Using standard identities, the Boolean operations $(\cdot, \vee,\oplus,\dots)$ on $f,g$ %
can be translated into $(\cdot,+)$ formulas, i.e., linear combinations of $f$ and $g$. 
Linearity %
then gives the resulting Hamiltonian in terms of $H_f$ and $H_g$. 
Explicitly, for the complement of a function $\overline{f}$, as  $\overline{f} = 1 -  f$, we have $H_{\overline{f}} = I - H_f$. Similarly, %
the %
logical identities %
$f\wedge g = fg$, $\; f \vee g = f + g -fg$, $\; f \oplus g = f + g -2fg$, and $f \Rightarrow g = \overline{f} + fg$, respectively, imply the 
remaining results of the theorem.  %
\end{proof} 

We %
summarize the Hamiltonian representations of several basic %
Boolean functions in Table~\ref{tab:basicHams} above, which %
are easily derived from Theorem~\ref{thm:1booleanFourierRep}. 
Applying the laws of %
Theorem~\ref{thm:compositionRules} 
we may derive Hamiltonians representing more complicated Boolean formulas than those of Table~\ref{tab:basicHams}, such as 
expressions with %
arbitrary numbers of variables, 
mixed types of clauses, 
or given as Boolean circuits. 
Some typical examples of Boolean functions on $3$ variables 
are the 
Majority ($MAJ$), Not-All-Equal ($NAE$), and $1$-in-$3$ functions, which behave as their names indicate. %
The Mod$_3$ function is $1$ when the sum $x_1 + x_2 + x_3$ is divisible by $3$, and 
satisfies Mod$_3 = \overline{NAE}$.  
We show the Hamiltonians representing these functions in Table~\ref{tab:3qubitHams}, 
which %
may be derived 
using either the composition rules of Theorem~\ref{thm:compositionRules} or the Fourier expansion approach of Theorem~\ref{thm:1booleanFourierRep}.
For example, %
$H_{1in3}$ follows %
applying Thm.~\ref{thm:compositionRules} %
with the identity 
$1in3(x_1,x_2,x_3)=x_1\overline{x}_2\overline{x}_3 + \overline{x}_1x_2 \overline{x}_3 +\overline{x}_1\overline{x}_2 x_3$.

The rules of Theorem \ref{thm:compositionRules} may be %
applied \textit{recursively}
to construct Hamiltonians representing %
more complicated Boolean functions, 
corresponding e.g. to parentheses in logical formulas, or wires in Boolean circuits. 
For example, the Hamiltonian representing the Boolean clause $f \vee g \vee h=f \vee (g \vee h)$ is given by
$H_{f \vee g \vee h} = 
H_f + H_{g\vee h}  -H_fH_{g\vee h}$, 
which simplifies to  $$H_{f \vee g \vee h} =H_f + H_g +H_h - H_fH_g - H_fH_h - H_gH_h + H_fH_gH_h.$$ 
Another example is the Majority function which satisfies
$$H_{MAJ(f,g,h)}=-2H_fH_gH_h + H_fH_g + H_fH_h + H_gH_h.$$

Together, the rules of Theorem \ref{thm:compositionRules} and Table \ref{tab:basicHams} show how 
to construct Hamiltonians representing functions given as arbitrary Boolean algebra~($\wedge,\vee$) or Boolean ring~($\cdot,\oplus$) elements, which are functionally complete in the sense of representing all possible Boolean functions~\cite{halmos2009boolean}.%

\subsection{Pseudo-Boolean Functions and Constraint Satisfaction Problems} \label{sec:pseudoBool}
Real functions on $n$ bits 
 are similarly represented as diagonal Hamiltonians via their Fourier expansion. 
Every such function $f\in\mathcal{R}_n$ may be expanded (non-uniquely) as a weighted sum of Boolean functions, possibly of exponential size. By linearity of the Fourier transform,  
the Hamiltonian $H_f$ is given precisely by the corresponding weighted sum of the Hamiltonians representing the Boolean functions. Moreover, 
the Hamiltonian $H_f$ is unique, so 
different expansions of $f$ as sums of Boolean functions must all result in the same %
 $H_f$.   

The Fourier coefficients %
are again given by the inner product (\ref{eq:innerProduct}) with the parity functions $ \chi_S$, 
\begin{equation}
\widehat{f} (S) = \langle f, \chi_S \rangle = \frac{1}{2^n} \sum_{x\in \{0,1\}^n} f(x) \chi_{S}(x), 
\end{equation}
and satisfy Parseval's identity as stated in~(\ref{eq:parseval}).
We are particularly interested in \textit{pseudo-Boolean functions}  
given as a weighted sum of logical clauses 
\begin{equation}
f(x) = \sum_{j=1}^m  w_j f_j (x),
\end{equation}
where $f_j \in \mathcal{B}_n$ and $w_j \in \reals$. %
Note that we do not deal explicitly with how the real numbers $w_j$ are represented and stored; for many applications %
they are bounded rational numbers and this issue is relatively minor; see, e.g., the constructions in  \cite{hadfield2019quantum,LucasIsingNP}. 
Indeed, in a constraint satisfaction problem, 
 typically all $w_j = 1$ and hence $f(x)$ gives the number of satisfied clauses (constraints).  

We have the following result which extends the previous %
results for Boolean functions. %

\begin{proof}[Proof of Theorem~\ref{thm:pseudoBool}]
By the linearity of the Fourier expansion and Theorem~\ref{thm:1booleanFourierRep} we have that an $n$-bit %
real function $f:\{0,1\}^n \rightarrow \reals$ is represented as the Hamiltonian  
$H_f = \sum_{S\subset [n]} \widehat{f} (S) \; \prod_{j\in S} Z_j$ with $\widehat{f} (S) = \langle f , \chi_s \rangle = \frac1{2^n}{\rm tr}( H_f \prod_{j\in S} Z_j)\, \in \reals.$ For pseudo-Boolean %
functions $f=\sum_{j=1}^m w_j f_j$ the bounds $d:= \deg(H_f) \leq  \max_j \deg(f_j)$ and 
$\;{\rm size}(H_f)\leq \min\{ \sum_j {\rm size}(H_{f_j})
, (e/d)^{d-1}n^d +1\}$ follow from simple counting as in the proof of  Theorem~\ref{thm:1booleanFourierRep}.
\end{proof}

Thus the results of Theorems~\ref{thm:1booleanFourierRep} and~\ref{thm:compositionRules} for Boolean functions also apply to the construction of Hamiltonians representing real functions, though with several important distinctions.   
Various %
subclasses of so-called pseudo-Boolean functions with particular %
attributes (e.g., submodularity) %
are important in applications 
and their properties may be further studied; see %
\cite{boros2002pseudo} for an overview. 

\begin{rem}  
In contrast to Theorem \ref{thm:1booleanFourierRep}, 
for a constraint satisfaction problem $f=\sum_{j=1}^m  f_j$, with $f_j\in \mathcal{B}_n$, %
applying Parseval's identity~(\ref{eq:parseval}) we have 
$$ \sum_{S\subset [n]}  \widehat{f} (S)^2 
 = \mathbf{E}[f] + 2\sum_{i<j} \langle f_i, f_j \rangle
 = \widehat{f}(\emptyset) + 2 \sum_{i<j} \mathbf{E}[f_i \wedge f_j] \geq \mathbf{E}[f] , $$
where $\mathbf{E}[f]:= \frac{1}{2^n}\sum_{x\in\{0,1\}^n} f(x)$ gives the expected value of $f(x)$ over the uniform distribution. %
In particular, $\sum_{S\subset [n]}  \widehat{f} (S)^2 
 = \mathbf{E}[f]$ if and only if $\langle f_i, f_j \rangle = 0$ for all $i,j$.  
 If there does exist an $i,j$ such that  $\langle f_i, f_j \rangle = 0$, then 
 the conjunction of the clauses is unsatisfiable, i.e. $\wedge_j f_j = 0$. 
\end{rem}

\section{Applications to Constructing Diagonal Operators} \label{sec:diagApplications}
Here we overview four immediate applications of our results to diagonal operators and problem encodings: constrained optimization, and ground state logic, simulating diagonal Hamiltonians, and quadratic optimization.

\subsection{Constrained Optimization} \label{sec:penalties}
In %
constrained optimization problems, we seek an optimal solution subject to satisfying a set of \textit{feasibility} constraints. These constraints may arise as part of the problem itself or from its encoding~\cite{LucasIsingNP,hadfield2019quantum}. %
In quantum annealing, 
a common approach to dealing with problem constraints
 is to augment the Hamiltonian that represents the objective function to be minimized with additional diagonal Hamiltonian terms that penalize (i.e., shift the eigenvalues of) states outside of the feasible subspace \cite{mcgeoch2014adiabatic,LucasIsingNP}.

Suppose we are given %
a constrained real function $f(x)$ %
to minimize, with a set of Boolean hard constraint functions~$g_j$, $j=1,\dots,\ell$, such that 
at least one of the $g_j(x) = 1$ if 
$x$ is an infeasible solution, and so  $\sum_j g_j(x) = 0$ if and only if $x$ is feasible. 
(Here, we may always redefine $f$ as to take a convenient value on infeasible states.) %
We may construct the  augmented problem Hamiltonian as
$$ H_p = H_f + \sum_{j=1}^\ell w_j H_{g_j},$$
where the Hamiltonians $H_f$ and $H_{g_j}$ represent $f$ and the $g_j$ as in Theorems \ref{thm:pseudoBool} and \ref{thm:1booleanFourierRep}, respectively. 
The $w_j$ are positive weights 
which may be selected appropriately such that 
infeasible basis states are eigenvectors of $H_p$ with eigenvalues  shifted away from those of $f$  (e.g., $w_j > \max(x) f(x)$), 
and feasible states are eigenvectors with eigenvalues $f(x)$. 
Hence, the ground state subspace of $H_p$ is spanned by states representing optimal feasible problem solutions. 
Theorems \ref{thm:compositionRules}, \ref{thm:1booleanFourierRep}, and \ref{thm:pseudoBool}  
may be applied to the functions $g_j$ to construct the penalty terms $H_{g_j}$ just as for the cost term $H_f$. 

Thus, our results may also be applied to explicitly construct problem mappings with penalty terms for constrained problems. See \cite{LucasIsingNP} for a number of specific problem mappings; we emphasize our approach may be applied directly to problems not considered therein. Furthermore, similar ideas apply to related approaches for constrained optimization such as %
Lagrange multipliers~\cite{ohzeki2020breaking}.

An alternative approach to constrained optimization %
is to map the hard constraint functions to diagonal Hamiltonians, which %
may be used to design mixing Hamiltonian that preserves the subspace of feasible states. %
This approach is proposed in constraint-preserving 
generalizations of quantum annealing~\cite{Hen2016driver,Hen2016quantum} and QAOA~\cite{hadfield2017qaoaPMES,hadfield2019quantum}, and can offer %
advantages %
over penalty-term approaches;  
see %
\cite{Hen2016driver,Hen2016quantum,hadfield2017qaoaPMES,hadfield2019quantum} for details and %
example  %
problem mappings.

\subsection{Ground State Boolean Logic}
With a universal quantum gate-model computer, a single (additional) control bit suffices for all efficiently computable control functions~\cite{NC}. Indeed, in principle we can always compute a Boolean function $f(x)$ in a control register %
by constructing a unitary operator $U_f:\ket{x}\ket{a} \rightarrow \ket{x}\ket{a\oplus f(x)}$, as we %
consider in Section \ref{sec:computeReg}. %
Nevertheless, in certain models or applications it is desirable to have a
purely Hamiltonian implementation. %

A different approach %
to computing a Boolean function $f\in \mathcal{B}_n$ in a register is to encode 
its input-output pairs as the ground state subspace of a Hamiltonian $H$,
referred to as \textit{ground state Boolean logic}; see, e.g., \cite{biamonte2008nonperturbative,crosson2010making,whitfield2012ground}. 
There are %
different ways to encode a function as the ground state subspace, %
as in there is freedom in how  the Hamiltonian acts on invalid computational basis states. 
For example, the function $AND(x,y)=xy$ can be encoded with the %
subspace $span\{\ket{x}\ket{y} \ket{xy}\} = span\{\ket{000},\ket{010},\ket{100},\ket{111}\}$,  %
which corresponds to the ground state subspace of a number of Hamiltonians. To construct such a Hamiltonian, the penalty term approach of the previous section could be used to penalize the invalid states. 
An alternative construction that takes advantage of the Hamiltonian representation~$H_f$ is to directly implement the Boolean function $g\in\mathcal{B}_{n+1}$ that satisfies  
$g(x,y)=0$ if and only if $y=f(x)$,  
or equivalently %
$g(x,y)=\overline{f(x) \oplus y}$. 
Applying Theorems~\ref{thm:1booleanFourierRep} and~\ref{thm:compositionRules} to this function and simplifying gives %
the following result. %

\begin{prop}
Let $f\in\mathcal{B}_n$ be represented by the %
Hamiltonian 
$H_f$ as in Theorem \ref{thm:1booleanFourierRep}, and $x_a:=\tfrac12 I -\tfrac12Z_a$ where the ancilla qubit is labelled $a$. 
Then the $(n+1)$-qubit Hamiltonian $H_g$ 
\begin{equation}    \label{eq:gsBooleanLogic}
H_g = I \otimes x_a +  H_{f} \otimes Z_a 
\end{equation}
represents the Boolean function $g \in \mathcal{B}_{n+1}$
which satisfies $g(x,y)=0$ if and only if $y=f(x)$, and 
has ground state subspace given by
$$span\{\ket{x}\ket{f(x)}: x\in \{0,1\}^n\}.$$ 

\end{prop}
Simpler Hamiltonians with the same ground state subspace may be found for specific classes of Boolean functions $f$; see, e.g., \cite{bian2010ising}. 
An advantage of the construction (\ref{eq:gsBooleanLogic}) is that it %
applies generally. %

\subsection{Simulating Diagonal Hamiltonians} 
\label{sec:HamSim}
Here we explain how Theorems \ref{thm:1booleanFourierRep} to \ref{thm:pseudoBool} may be straightforwardly applied to yield quantum circuits 
implementing time evolution under diagonal Hamiltonians, i.e., diagonal unitaries, %
which is a fairly well-studied problem with  %
several related approaches proposed in the literature.

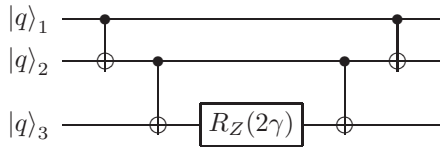
\begin{figure}[h] %
\centerline{
\Qcircuit @C = 1.2 em @R = 1.2 em {
\lstick{\ket{q}_1} & \ctrl{1} & \qw      & \qw                             & \qw       & \ctrl{1}  & \qw \\
\lstick{\ket{q}_2} & \targ    & \ctrl{1} & \qw                              &  \ctrl{1} & \targ    & \qw \\ 
\lstick{\ket{q}_3} & \qw      & \targ    & \gate{R_Z (2\gamma) } & \targ     &  \qw    & \qw \\ 
}
}
\caption{Quantum circuit performing the operation $U=exp(-i\gamma Z_1 Z_2 Z_3)$ on %
three qubits labeled $1$, $2$, and $3$. The middle operator is a $Z$-rotation gate, and the other gates are controlled-NOT (CNOT) gates with
a black circle indicating the control qubit and cross indicating the target. 
By similar circuits, %
$U=exp(-i\gamma Z_1 Z_2 \dots Z_\ell)$ can be implemented with $2(\ell-1)$ CNOT gates and one $R_Z$ gate. 
Different circuit compilations are possible, including compilation to different gate sets. 
} 
\label{fig:RZZcircuit}
\end{figure}

In many applications we desire to \textit{simulate} a Hamiltonian~$H$ for some time $\gamma\in\reals$, i.e., implement exactly or approximately the unitary operator 
$ U(\gamma)= e^{-i\gamma H}.$ 
When $H$ is diagonal, efficient quantum circuits may be obtained using Fourier analysis~\cite{schuch2003programmable,welch2014efficient,amy2018controlled} or other approaches.  
Consider the simulation of a Hamiltonian $H_f$ representing a real or Boolean function $f$. 
It is well known that if %
$f$ can be efficiently computed classically, and if ancilla qubits are available, then the Hamiltonian $H_f$ can be simulated efficiently by computing $f$ in a scratchpad register and performing a sequence of controlled rotations; see, e.g., \cite{childs2004quantum}. 
These methods typically avoid computing the Fourier expansion of $f$ explicitly. %
On the other hand, there exist applications where an explicit Hamiltonian-based implementation is desirable, 
such as quantum annealing, or cases where we wish to minimize the need for ancilla qubits, such as, for example, near-term %
implementations. %
Efficient circuits simulating products of Pauli $Z$ operators are well-known \cite{childs2004quantum,barenco1995elementary}, as 
shown in Figure~\ref{fig:RZZcircuit}. %
As Pauli $Z$ terms mutually commute, circuits simulating individual terms 
in the Hamiltonians (\ref{eq:Hfexpan}) or  (\ref{eq:Hfexpan2})
can be applied in any sequence as to simulate their sum. %
Thus, %
when ${\rm size}(H_f)=O({\rm poly}(n))$ %
we can always simulate $H_f$ efficiently in this way. 
We summarize this observation in the following. %

\begin{cor}
A Hamiltonian $H_f$ 
representing a Boolean or real function~$f$ as in (\ref{eq:Hfexpan}) or (\ref{eq:Hfexpan2}) 
can be simulated, i.e., the operation $exp(-i\gamma H_f)$ implemented, %
with 
$n$ qubits and 
$O(\deg(H_f)\cdot \text{size}(H_f))$ 
basic quantum gates. %
In particular, Hamiltonians~$H_f$ with bounded maximum degree~$d:=\deg(H_f)=O(1)$ can be simulated with~$O(n^d)$ basic gates. Ancilla qubits are not necessary in either case. 
\end{cor}

Here, by \textit{basic quantum gates} we mean the set of CNOT and single qubit rotation gates, which is a 
standard universal set \cite{barenco1995elementary,NC}.  
We remark that the Hamiltonian simulation considered in the %
corollary is exact in the sense that if each of the %
basic gates %
is implemented exactly, then so is $exp(-i\gamma H_f)$. The approximation of quantum gates and operators is an important topic but we do not deal with it here; see, e.g., \cite{NC}.

\begin{example} \label{ex:grovers}[Application to Grover's algorithm] %
For a %
Boolean function $f$, %
simulating $H_f$ for time $\pi$ gives the standard oracle query for Grover's algorithm \cite{NC}
\begin{equation}  \label{eq:GroversQuery}
e^{-i \pi H_f} \ket{x} = (-1)^{f(x)} \ket{x} .
\end{equation}
Hence, %
when $H_f$ is known explicitly and 
$ \text{size}(H_f)={\rm poly}(n)$,  %
we can %
efficiently construct and implement the operator $(-1)^{f(x)} $  
using quantum circuits for simulating $H_f$ using only CNOT and $R_Z$ gates, without any necessary ancilla qubits. 
\end{example}

\subsection{Quadratic Unconstrained Binary Optimization}
A general and important class of pseudo-Boolean optimization problems are 
\textit{quadratic unconstrained binary optimization} (QUBO) problems \cite{mcgeoch2014adiabatic}, where we seek to maximize or minimize a
degree-two pseudo-Boolean 
function 
\begin{equation}  \label{eq:QUBO}
f(x) = a + \sum_{j=1}^n c_j x_j + \sum_{j<k} d_{jk} x_{j} x_k,
\end{equation}
with $a, c_j, d_{jk}\in \reals$ and $x_j \in \{0,1\}$. %
Indeed, this is the class of problems %
(ideally) implementable on current quantum annealing devices such as, for example, %
D-WAVE machines, where the qubit interactions are themselves quadratic \cite{bian2010ising,mcgeoch2014adiabatic}.  
The QUBO class also contains many problems which at first sight are not quadratic, 
 via polynomial reductions which often require extra variables; %
indeed, %
the %
natural 
QUBO decision problem is %
NP-complete~\cite{GareyJohnson}. 
Note that $\overline{x}_j=1-x_j$, so (\ref{eq:QUBO}) is without loss of generality. Applying our above results gives the following. 

\begin{cor}
The QUBO objective function (\ref{eq:QUBO}) maps to a Hamiltonian given as a quadratic %
sum of Pauli $Z$ operators, with size($H_f) \leq 1 + n/2 + n^2/2$. 
Explicitly, we have
\begin{equation} \label{eq:QUBOHam}
H_f = (a + c + d ) I - \frac12 \sum_{j=1}^n (c_j + d_j) Z_j + \frac14 \sum_{j<k} d_{jk} Z_{j} Z_k,
\end{equation}
where we have defined $c = \frac12 \sum_{j=1}^n c_j, d = \frac14 \sum_{j<k} d_{jk}$, and $d_j = \frac12 \sum_{k:k\neq j} d_{jk}$ with $d_{jk}=d_{kj}$.

Moreover, we can simulate $H_f$, i.e., implement the phase operator $U_p(t)=e^{-it H_f}$, using at most $n$ many $R_Z$ rotation gates and $\binom{n}{2}$ many $R_{ZZ}$ gates.
\end{cor}

The QUBO problem is closely related to the Ising model of interacting spins from physics~\cite{nishimori2001statistical}.  

\begin{example}[Interacting Ising spins]
Consider the related classical ISING decision
problem of %
determining 
whether the ground state energy (lowest eigenvalue) of an Ising model (degree two) diagonal Hamiltonian 
$$ H = a_0 I + \sum_j a_j Z_j + \sum_{j<k} a_{jk} Z_jZ_k  $$ 
is at most a given constant. %
As the number of terms %
is ${ size}(H)=O(n^2)$, 
we can efficiently check the energy $H(x)$ of each candidate ground state $\ket{x}$, $x\in\{0,1\}^n$, so 
the problem is in NP. On the other hand, from Theorem \ref{thm:compositionRules}, we see that the NP-complete problem MAX-$2$-SAT maps to a degree-two Hamiltonian of this form, with solution encoded in the ground state energy of $-H$. 
Thus, from our results it trivially follows 
that ISING is NP-complete. 
Similar arguments can be used to show NP-completeness with 
restricted coefficients values %
(e.g., antiferromagnetic)
or with restricted interaction topologies such as planar %
graphs~\cite{barahona1982computational}.  %
\end{example}

\section{Applications to non-diagonal operators}
\label{sec:nonDiagHams}
Here we consider the application of our results beyond strictly diagonal operators, in particular to constructing controlled unitaries (quantum gates) common in quantum algorithms. 

Just as we have seen that diagonal Hamiltonians correspond to linear combinations of Pauli $Z$ operators, 
tensor products of arbitrary single-qubit Pauli matrices give a basis for the vector space of $n$-qubit Hamiltonians. 
Indeed, a general Hamiltonian $H$ may be expanded as a sum of Pauli matrices as 
\begin{equation}  \label{eq:PauliExpansion}
H = a_0 I + \sum_{j=1}^n \sum_{\sigma = X,Y,Z} a_{j\sigma} \sigma_j 
+ \sum_{j \neq k } \sum_{\sigma = X,Y,Z} \sum_{\lambda = X,Y,Z}  a_{jk\sigma\lambda} \sigma_j  \lambda_k + \dots ,
\end{equation}
with real coefficients %
$a_\alpha \in \reals$. 
(For general linear operators on qubits the same formula holds but with $a_\alpha \in \complex$, see e.g. %
~\cite{woitBook}.) 
The coefficients are easily shown to satisfy
\begin{equation}  \label{eq:PauliTrace}
a_{\alpha} = \frac{1}{2^n} {\rm tr}(\alpha H),    %
\end{equation}
for each of the $4^n$ Pauli terms $\alpha=\sigma_1 \sigma_2 \dots \sigma_n$, %
$\sigma_j \in \{I,X_j,Y_j,Z_j\}$, 
which are orthonormal with respect to the Hilbert-Schmidt inner product $\langle \alpha,\beta\rangle := \tfrac1{2^n}{\rm tr}(\alpha^\dagger \beta)$ that generalizes~(\ref{eq:innerProduct}).

\subsection{Controlled Unitaries and Hamiltonians} 
\label{sec:ControlledHams}
Many quantum algorithms require controlled Hamiltonian evolutions. 
For example, in quantum phase estimation (QPE) \cite{NC}, we require transformations on $(1+n)$-qubit basis states of the form 
$$ \ket{0}\ket{x} \rightarrow \ket{0}\ket{x}, \;\;\;\;\;\;\;\;\;\; \ket{1}\ket{x} \rightarrow  \ket{1}e^{-iHt}\ket{x} ,$$
for various values %
$t=1,2,4,\dots$. 
Consider such a transformation with fixed $t$. 
Labeling the first register (control qubit) $a$, 
the overall unitary %
may be written as
\begin{equation}  \label{eq:cphaseoracle}
\Lambda_{x_a} ( e^{-iHt})  = \ket{0}\bra{0} \otimes I + \ket{1}\bra{1} \otimes e^{-iHt}.
\end{equation}
Recall the notation $ \Lambda_{x_a} ( e^{-iHt})$ indicates the unitary $e^{-iHt}$ controlled by the classical function $x_a$. 
We obtain the Hamiltonian corresponding to this transformation %
by writing %
$\Lambda_{x_a} ( e^{-iHt})=e^{-i\widetilde{H}t}$, which gives
\begin{equation}
 \widetilde{H} = \ket{1}\bra{1} \otimes H= x_a \otimes H = \frac12 I \otimes H - \frac12 Z_a \otimes H.
\end{equation}
Recall that for simplicity we %
sometimes write a function $f$ %
in place of its Hamiltonian representation $H_f$ %
as we have done here for the function  $f(x) = x_a$. 
Note that the control qubit is assumed precomputed here; its value may or may not depend on $x$. 

More generally, consider Hamiltonian evolution controlled by a Boolean function $g\in \mathcal{B}_k$ acting on a $k$-qubit ancilla register. 
In this case we seek to affect the unitary transformation on $(k+n)$-qubit basis states 
$$ \ket{y}\ket{x} \rightarrow \ket{y}\ket{x} \;\;\;\;\;\;\;\;\;\;\;\;\;\;\; \text{if } g(y) = 0,
$$
$$ \ket{y}\ket{x} \rightarrow  \ket{y}e^{-iHt}\ket{x}
\;\;\;\;\;\;\;\;\;\; \text{if } g(y) = 1,$$
which gives the overall unitary
\begin{equation}  \label{eq:conrolUnitary}
 \Lambda_g ( e^{-iH t}) = \sum_{y: g(y)=0} \ket{y}\bra{y} \otimes I + \sum_{y: g(y)=1} \ket{y}\bra{y} \otimes e^{-iH t} = H_{\overline{g}} \otimes I + H_g \otimes e^{-iH t},
 \end{equation}
corresponding to evolution under the Hamiltonian 
\begin{equation}   \label{eq:controlledHam} 
 \widetilde{H}_g = \sum_{y: g(y)=1} \ket{y}\bra{y} \otimes H = H_g \otimes H . 
\end{equation}
These results have been summarized in Proposition \ref{prop:controlledHams} above. 
Note that if the Hamiltonian $H=H_f$
represents a Boolean function $f$, then (\ref{eq:controlledHam}) represents the conjunction of $f$ and $g$, i.e., we have $\widetilde{H}_g = H_g \otimes H_f = H_{g\wedge f}$. 

Moreover, for an arbitrary Hamiltonian $H$, if we can implement the ancilla controlled operator $\Lambda_{x_a} (e^{-iHt})$
for sufficiently many values of $t \in \reals$, then it is straightforward to %
implement a variable-time controlled Hamiltonian simulation operator $\Lambda_{\tau,H}$, 
which acts on basis states as
\begin{equation}
\Lambda_{\tau,H} \ket{\tau}\ket{x}= \ket{\tau}e^{-iH\tau}\ket{x}.
\end{equation}
For example, if $\tau$ was encoded in binary, then  $\Lambda_H$ could be implemented by applying the operators $\Lambda_{x_{a_0}}(e^{-iH2^0})$, $\Lambda_{x_{a_1}}(e^{-iH2^1})$, $\dots$, $\Lambda_{x_{{a_j}}}(e^{-iH2^j})$, $\dots$ in sequence controlled over each bit $\tau_j$ in the first register; see, e.g., \cite{NC}. 

Finally, we remark on the condition %
that the control function and target unitaries act on disjoint sets of qubits, and the relation to spin creation and annihilation operators. 
\begin{example}[Creation/annihilation operators and interacting fermions]  \label{rem:spinOps}
Consider %
the operator $Xx$ which acts as $Xx\ket{0}=0$ and $Xx\ket{1}=\ket{0}$. 
Clearly $Xx=X(I-Z)/2=X/2 +iY/2$ is not self-adjoint. 
Indeed, %
 $Xx$ is equivalently the well-known \textit{spin annihilation} operator $b=\ket{0}\bra{1}$. 
The \textit{spin creation} operator~$b^\dagger$ is similarly given as $b^\dagger=(Xx)^\dagger = xX = X \overline{x}=\ket{1}\bra{0}$, and the spin number (occupation) operator is $b^\dagger b = x=(I-Z)/2$. %
Thus, the product of two Hamiltonians $H_g$ and $H$ acting nontrivially on overlapping qubits is not guaranteed to yield a self-adjoint operator as in~(\ref{eq:controlledHam}). 

For systems of fermions in the second quantized (occupation number) representation \cite{Szabo}, 
Hamiltonians consist of %
polynomials of \textit{fermionic creation and annihilation operators}~$a^\dagger$ and $a$, %
satisfying the \textit{canonical anticommutation relation} algebra $\{a_j,a_k\}=\{a^\dagger_j,a^\dagger_k\}=0$ and $ \{a_j,a^\dagger_k\}=\delta_{jk}$. 
It is easily shown that these relations are satisfied if we %
use spin operators and the parity functions $\chi_S$, $S\subset [n]$ of (\ref{eq:parityfcn}) to represent the fermionic operators as $ a_j = \chi_{\{1,2,\dots, j-1\}}b_j$. %
Hence, from above we may represent these operators as
 $a_j = Z_1 Z_2 \dots Z_{j-1} (X_j+iY_j)/2$ %
  and $a^\dagger_j = Z_1 Z_2 \dots Z_{j-1} (X_j-iY_j)/2$
which reproduces the well-known \textit{Jordan-Wigner transform}.  
This representation is used in many applications, such as quantum algorithms for quantum chemistry, though alternative mappings %
with different tradeoffs are known~\cite{seeley2012bravyi}.  
\end{example}

\subsection{Computing Functions in Registers} \label{sec:computeReg} %
Here we consider the special case of computing a Boolean function in a register and show a similar %
complexity result as for the diagonal Hamiltonian case. As an illustrative example we briefly consider the connection of our results to quantum query complexity and Grover's algorithm.  

Suppose for a Boolean function $f$ we have a unitary operator $G_f$ that acts on each 
$(n+1)$-qubit basis state $\ket{x}\ket{a}$, $a\in \{0,1\}$ as 
\begin{equation} \label{eq:Gf}
G_f \ket{x}\ket{a} = \ket{x}\ket{a\oplus f(x)}.
\end{equation}
If we let the function $f$ be arbitrary and unknown, then each application of $G_f$ (considered a black-box) is called an \textit{oracle query} for $f$. Operators $G_f$ as in (\ref{eq:Gf}) are called \textit{bit queries}. 
Note that $G_f$ may often be derived from a reversible classical circuit for computing $f$, but we consider it abstractly here. (%
The query (\ref{eq:GroversQuery}) for Grover's algorithm %
is often referred to as a \textit{phase query}.) 

We show that $G_f$ also induces a diagonal representation of $H_f$, and hence faces the same computational difficulties for the Boolean case. 
Observe that the diagonal Hamiltonian
\begin{equation}
    H'_f \,:=\, G_f x_a G_f \,=\, \frac12 I - \frac12 G_f Z_a G_f
= x_a + H_f Z_a
\end{equation}
acts on computational basis states as
\begin{equation}  \label{eq:Hfprime}
H'_f \ket{x}\ket{a} = (f(x)\oplus a) \ket{x}\ket{a},
\end{equation}
so we identify
$H'_f=H_{f\oplus a}$. Hence when $a=0$, $H'_f$ gives an $n+1$ qubit representation of $f$.

\begin{proof}[Proof of Prop. \ref{prop:computeReg}] %
The first statement %
follows from applying Prop.~\ref{prop:controlledHams} with $H_f$ as given from Thm.~\ref{thm:1booleanFourierRep} and $H=X_a$. For the second statement, we use (\ref{eq:PauliExpansion}), (\ref{eq:PauliTrace}), and 
(\ref{eq:Hfprime})
to expand $G_f$ as a sum of Pauli terms acting on $n$ qubits $G_f = (1-\widehat{f}(\emptyset))I + \widehat{f}(\emptyset)X_a + \dots$, where none of the terms to the right are proportional to~$I$. Hence %
the identity coefficient $\widehat{g}(\emptyset)={\rm tr}(G_f)/2^{n+1}=1-\widehat{f}(\emptyset)$ %
gives (one minus) the fraction of satisfying assignments of~$f$. 
\end{proof}

We conclude the section with a %
pedagogic example application to oracle models.

\begin{example}[Phase kickback and function queries]
Recall Example~\ref{ex:grovers}. For an arbitrary Boolean function~$f$, a single application of the \textit{bit-query oracle} $G_f$  %
suffices to simulate $(-1)^{f}$ as 
\begin{equation} \label{eq:oracle1}
    G_f  \ket{x}\ket{-}_a=(-1)^{f(x)}\ket{x}\ket{-}_a,
\end{equation} 
using a single ancilla qubit prepared in the state $\ket{-}=\frac{1}{\sqrt{2}} (\ket{0}-\ket{1})={\textrm H}\ket{1}$ and where ${\textrm H}$ denotes the Hadamard (single qubit quantum Fourier transform) gate~\cite{NC}, known as phase kickback. %

Next consider the \textit{controlled-phase-query oracle} $\Lambda_{x_a} ( (-1)^{f} )=\Lambda_{x_a} (e^{-i\pi {H_f}})=(-1)^{f \wedge x_a}$ derived from (\ref{eq:cphaseoracle}). Observe that if we can %
query $\Lambda_{x_a} ( (-1)^{f} )$ 
then we can implement $G_f$ using a single-bit \textit{quantum phase estimation}, which requires two Hadamard gates applied to an ancilla qubit and a single $\Lambda_{x_a} ((-1)^f)$ query as
\begin{equation}  \label{eq:BitToPhase}
{\textrm H}_a \, \Lambda_{x_a} ((-1)^f)  \,  {\textrm H}_a \ket{x}\ket{0}_a = \ket{x}\ket{0\oplus f(x)}_a = G_f \ket{x}\ket{0}_a.
\end{equation}
Similarly, two %
bit queries can simulate $\Lambda_{x_a} ((-1)^f)$  
using an %
ancilla qubit $\ket{0}_b$ to store $f(x)$ %
as 
\begin{equation} \label{eq:phaseToBit}
 (I\otimes G_f )\; 
 R_{Z_a}(-\tfrac{\pi}2) R_{Z_b}(-\tfrac{\pi}2) R_{Z_aZ_b}(\tfrac{\pi}2)
 \; (I\otimes G_f ) \; \ket{x_a}_a \ket{x}\ket{0}_b = c'(-1)^{f(x)\wedge x_a} \ket{x_a}_a \ket{x}\ket{0}_b ,
\end{equation}
where the constant $c'$ is an unimportant global phase. %
(Here the three rotations on the left of (\ref{eq:phaseToBit}) are derived from simulating the Hamiltonian representing the function $x_a\wedge x_b$ for time $\pi$ as to implement the operator $(-1)^{x_a\wedge x_b}$, and the second application of $G_f$ uncomputes the ancilla qubit.)
Hence, %
we easily see that the oracles  
$G_f$ and  $\Lambda_{x_a} (e^{-i \pi H_f})=(-1)^{f(x)\wedge x_a}$ are computationally equivalent, and both are at least as powerful as the phase oracle $(-1)^f$. 
The results of this paper may be similarly applied to the study of complexity and reductions between further classes of oracles.
\end{example}

\section{Discussion and Future Work}\label{sec:discussion}
We have %
shown explicit rules and results for constructing Hamiltonians representing Boolean and pseudo-Boolean functions, %
including important classes of 
objective functions for %
combinatorial optimization. 
 Applications include quantum gate-model and quantum annealing approaches to optimization, the simulation of diagonal Hamiltonians, and the construction of penalty or mixer  %
 Hamiltonians for problems with hard constraints. Moreover, we have shown our results entail quantum circuit constructions for controlled unitaries, %
 in particular, those that reversibly compute a Boolean function in an ancilla register.  
The goal of these results is to give a toolkit which can be used generally towards %
the design and implementation of a wide range of quantum approaches for %
optimization, for related applications such as machine learning \cite{verdon2017quantum}, and beyond. %
Our results give a unified %
view of existing problem mappings in the literature, such as those of \cite{hadfield2019quantum,LucasIsingNP}.

There are a variety of enticing applications and extensions of our results, and we briefly outline several additional directions.  %
We %
emphasize that 
Fourier analysis is generally a very rich topic in computer science, mathematics, and physics, 
in addition to its %
many applications in quantum computing. 
A promising research direction is to further apply these tools, in particular more advanced ideas from the Fourier analysis of Boolean functions as applied in classical computer science, to the design and analysis of quantum algorithms.
As mentioned, we leave %
a detailed analysis of general Hamiltonians acting on qubits or qudits as a topic of future work. 
A next step is to %
further classify \textit{unfaithful} Hamiltonian representations of Boolean functions, where $n$ variables are encoded in %
$n' > n$ qubits, which in particular is an important paradigm 
for embedding problems on physical quantum annealing hardware, 
and, more generally, is related to the theory of quantum error correcting codes. %
Moreover, 
our results %
may have useful applications to %
quantum statistical mechanics, where many important Hamiltonians are given as linear combination of Pauli operators,  %
such as the Ising or quantum XY models \cite{barahona1982computational,lieb2004two}. 
Furthermore, exploring connections to quantum complexity theory 
may be fruitful, %
such as %
Hamiltonian complexity \cite{KempeLocalHam,gharibian2015quantum,cubitt2017universal}, 
or the computational power of restricted classes of quantum circuits~\cite{knill1998power,
bremner2017achieving,bravyi2017quantum}. Finally, it of interest whether techniques from %
quantum computing and quantum information can shed further insight on important open problems in classical computer science~\cite{gu2019majorana,filmus2014real}.

\section*{Acknowledgments}
We thank Al Aho for support and guidance, and the members of the Quantum AI Lab for providing helpful comments and suggestions. This work was initiated thanks to the support of the USRA Feynman Quantum Computing Academy summer internship program and NASA Ames Research Center, and further developed while at Columbia University. S.H. was additionally supported by NASA Academic Mission Services, Contract No. NNA16BD14C. The views and conclusions contained herein are those of the author and should not be interpreted as necessarily representing the official policies or endorsements, either expressed or implied, of the U.S. Government. The U.S. Government is authorized to reproduce and distribute reprints for Governmental purpose not withstanding any copyright annotation thereon.

\bibliographystyle{ieeetr}
\bibliography{bib}

\end{document}